\documentclass[11pt, a4paper]{article}

\usepackage{authblk}
\usepackage{bm}
\usepackage{amssymb, amsmath}
\usepackage{color}
\usepackage[amsmath,thmmarks]{ntheorem}
\usepackage{array}
\newcolumntype{C}[1]{>{\centering}m{#1}}
\usepackage{color}
\usepackage{todonotes}
\usepackage{enumitem}
\usepackage{geometry}
\usepackage[ruled, vlined, linesnumbered]{algorithm2e}

\DeclareMathOperator{\ini}{ini}

\DeclareMathOperator{\lt}{lt}
\DeclareMathOperator{\lv}{lv}
\DeclareMathOperator{\pquo}{pquo}
\DeclareMathOperator{\prem}{prem}
\DeclareMathOperator{\sat}{sat}
\DeclareMathOperator{\res}{res}
\newcommand{\field}[1]{\mathbb{#1}}
\newcommand{\fk}{\field{K}}
\newcommand{\pset}[1]{\mathcal{#1}}
\newcommand{\ideal}[1]{\mathfrak{#1}}
\newcommand{\bases}[1]{\langle #1 \rangle}
\newcommand{\znum}{\mathbb{Z}}
\newcommand{\qnum}{\mathbb{Q}}
\newcommand{\grobner}{Gr\"{o}bner }
\newcommand{\p}[1]{\bm{#1}}
\newcommand{\kx}{\field{K}[\p{x}]}
\DeclareMathOperator{\zero}{\mathsf{Z}}
\DeclareMathOperator{\nform}{nform}
\DeclareMathOperator{\algnor}{\sf CharDec}

\newtheorem{theorem}{Theorem}[section]
\newtheorem{proposition}[theorem]{Proposition}
\newtheorem{lemma}[theorem]{Lemma}
\newtheorem{corollary}[theorem]{Corollary}

\theorembodyfont{\rmfamily}
\newtheorem{remarks}[theorem]{Remark}
\newtheorem{example}[theorem]{Example}
\newtheorem{definition}[theorem]{Definition}

\theoremstyle{nonumberplain}
\theoremsymbol{$\square$}
\newtheorem{proof}{Proof}
\begin{document}

\title{Decomposition of polynomial sets into characteristic pairs\footnote{This work was supported partially by the National Natural Science Foundation of China (NSFC 11401018) and the project SKLSDE-2015ZX-18.}}

\date{}

\author[ab]{Dongming Wang}
\author[a]{Rina Dong}
\author[a]{Chenqi Mou\footnote{Corresponding author: +8613811426823, chenqi.mou@buaa.edu.cn, School of Mathematics and Systems Science, Beihang University, Beijing 100191, China. }}
\affil[a]{LMIB -- SKLSDE -- School of Mathematics and Systems Science, \authorcr
Beihang University, Beijing 100191, China \authorcr
\{rina.dong, chenqi.mou\}@buaa.edu.cn \vspace{4mm}}

\affil[b]{Centre National de la Recherche Scientifique, \authorcr
75794 Paris cedex 16, France \authorcr dongming.wang@lip6.fr}

\maketitle
\begin{abstract}
A characteristic pair is a pair $(\pset{G}, \pset{C})$ of polynomial sets in which $\pset{G}$ is a reduced lexicographic \grobner basis, $\pset{C}$ is the minimal triangular set contained in $\pset{G}$, and $\pset{C}$ is normal. In this paper, we show that any finite polynomial set $\pset{P}$ can be decomposed algorithmically into finitely many characteristic pairs with associated zero relations, which provide representations for the zero set of $\pset{P}$ in terms of those of \grobner bases and those of triangular sets. The algorithm we propose for the decomposition makes use of the inherent connection between Ritt characteristic sets and lexicographic \grobner bases and is based essentially on the structural properties and the computation of lexicographic \grobner bases. Several nice properties about the decomposition and the resulting characteristic pairs, in particular relationships between the \grobner basis and the triangular set in each pair, are established. Examples are given to illustrate the algorithm and some of the properties.
\end{abstract}

\noindent{\small {\bf Key words: }Characteristic pair, normal triangular set, lexicographic \grobner basis, zero decomposition\\

\noindent{\small {\bf Mathematics Subject Classification: }68W30 (primary), 13P10, 13P15 (secondary)

\section{Introduction}
\label{sec:intro}
Systems of polynomial equations are fundamental objects of study in mathematics which occur in many domains of science and engineering. Such systems may be triangularized by using the well-known method of Gaussian elimination when the equations are linear. There are two approaches, developed on the basis of characteristic sets \cite{r50d,Wu94m} and \grobner bases \cite{B85G,CLO1997I}, which can be considered as generalizations of Gaussian elimination to the case where the equations are nonlinear. Following these approaches of triangularization, the present paper is concerned with the problem of decomposing an arbitrary set $\pset{P}$ of multivariate polynomials into finitely many triangular sets of polynomials that may be used to represent the set of zeros of $\pset{P}$ (or equivalently the algebraic variety defined by $\pset{P}$, or the radical of the ideal generated by $\pset{P}$). This problem of triangular decomposition is conceptually simple, but computationally difficult, and to it satisfactory algorithmic solutions are of both theoretical interest and practical value. The last three decades have witnessed extensive research on polynomial elimination and triangular decomposition, which led to significant developments on the theories, methods, and software tools for polynomial system solving (see, e.g., \cite{a99t,AM99t,B2011a,c07c,CM2012a,CG1990R,GC1992s,h03n,k93g,l91n,P2013o,w93e,w98d,w00c,w86z,Wu01m} and references therein). Along with these developments, triangular decomposition has become a standard approach to studying computational problems in commutative algebra and algebraic geometry, a basic toolkit for building advanced functions in modern computer algebra systems, and a general and powerful technique of breaking complex polynomial systems down into simply structured, easily manageable subsystems for diverse scientific and engineering applications.

To make our statements precise, we fix an order for the variables of the polynomials in question. A \emph{triangular set} $\pset{T}$ is meant an ordered set of polynomials whose greatest variables strictly increase with respect to the fixed variable order. $\pset{T}$ is said to be \emph{normal} or called a \emph{normal set} if none of the greatest variables occurs in the leading coefficients of the polynomials in $\pset{T}$ with respect to their greatest variables. By a \emph{polynomial system} we mean a pair $[\pset{P}, \pset{Q}]$ of polynomial sets with which the system of polynomial equations $\pset{P} = 0$ and inequations $\pset{Q} \neq 0$ is of concern. It is called a \emph{triangular system} or a \emph{normal system}, respectively, if $\pset{P}$ is a triangular set or a normal set and $\pset{Q}$ satisfies certain subsidiary conditions.

Effective algorithms are now available for decomposing arbitrary polynomial sets or systems of moderate size into triangular sets or systems of various kinds (regular, simple, irreducible, etc.) \cite{W2001E,W2004e,k93g,li08a}, though it is not yet clear how to measure the quality of triangular decompositions and how to produce triangular sets or systems of high quality in terms of theoretical properties (such as uniqueness, squarefreeness, and normality) and simplicity of expression (with lower degree, smaller size, and fewer components, etc.). One way to obtain ``good" triangular decompositions is via computation of lexicographic (lex) \grobner bases, where the lex term ordering determined by the variable order ensures that the bases have certain triangular structures with nice algebraic properties \cite{B1965A,CLO1997I}. For the zero-dimensional case, relationships between \grobner bases and triangular sets were studied in \cite{L1992S}, leading to algorithms for the computation of triangular sets from lex \grobner bases based on factorization and the D5 principle \cite{D1985a}. More recently, an algorithm for triangular decomposition of zero-dimensional polynomial sets has been proposed in \cite{D2012o}, based on an exploration of the structures of lex \grobner bases. For polynomial ideals of arbitrary dimension, the connection between Ritt characteristic sets and lex \grobner bases has been investigated in \cite{W2016o}; (pseudo-) divisibility relationships established therein will be clarified and used in later sections.
 The structures of lex \grobner bases were studied first by Lazard \cite{L85i} for bivariate ideals and then extended to general zero-dimensional (radical, multivariate) ideals in a number of papers \cite{K1987s,G89p,MM2003r,D2012o} with many deep results.

One kind of presumably good triangular sets is normal sets explained above, which appeared for the first time as \emph{normalized} triangular sets in \cite{l91n} and later as \emph{p-chains} in \cite{GC1992s}, and were elaborated in \cite[Sect.\ 5.2]{W2001E} and \cite{D2001o}. Normal sets and systems enjoy a number of remarkable properties and are convenient for various applications, in particular dealing with parametric polynomial systems \cite{c07c,GC1992s,W2001E}. There are algorithms for normalizing triangular sets, and more generally, for decomposing arbitrary polynomial sets or systems into normal sets or systems \cite{W2001E,WZ2007a,li08a}.

In this paper, we focus our study on what we call \emph{characteristic pair} and \emph{characteristic decomposition}: the former is a pair $(\pset{G}, \pset{C})$ of polynomial sets in which $\pset{G}$ is a reduced lex \grobner basis, $\pset{C}$ is the minimal triangular set contained in $\pset{G}$, and $\pset{C}$ is normal; the latter is the decomposition of a finite polynomial set $\pset{P}$ into finitely many characteristic pairs with associated zero relations, which provide representations for the zero set of $\pset{P}$ in terms of those of \grobner bases and those of triangular sets. Our main contributions include: (1) clarification of the connection between normal sets and lex \grobner bases via the concept of W-characteristic sets (introduced in \cite{W2016o}), (2) introduction of the concepts of (strong) characteristic pairs and characteristic decomposition with several properties proved, (3) an algorithm for computing (strong) characteristic decompositions of polynomial sets, and (4) experimental results illustrating the performance of our algorithm and its implementation.

The proposed algorithm, which makes use of the inherent connection between characteristic sets and \grobner bases for splitting, is capable of decomposing any given polynomial set simultaneously into finitely many normal sets $\pset{C}_1, \ldots, \pset{C}_t$ and lex \grobner bases $\pset{G}_1, \ldots, \pset{G}_t$ with every $\pset{C}_i$ contained in $\pset{G}_i$. It is proved that each $\pset{C}_i$ can be reduced to a Ritt characteristic set of the ideal generated by $\pset{G}_i$ if it is not reduced (Theorem~\ref{thm:reduced} and Corollary~\ref{cor:RittDec}). It is also shown that a strong characteristic decomposition can be computed out of any characteristic decomposition without need of further splitting (Theorems~\ref{prop:equiv} and \ref{thm:zeroStrong}).

After a brief review of \grobner bases, normal sets, and W-characteristic sets in Section~\ref{sec:pre}, we will define (strong) characteristic pairs and (strong) characteristic decomposition and prove some of their properties in Section~\ref{sec:prop}, describe the decomposition algorithm with proofs of termination and correctness in Section~\ref{sec:alg}, and illustrate how the algorithm works with an example and report our experimental results in Section~\ref{sec:ex-ex}.

\section{\grobner bases, W-characteristic sets, and triangular sets}
\label{sec:pre}

We recall some basic notions and notations which will be used in later sections and highlight the intrinsic structures of reduced lex \grobner bases on which the main results of this paper are based. For more details about the theories of \grobner bases (also called Buchberger-\grobner bases) and triangular sets, the reader is referred to \cite{a99t,CLO1997I,W2001E} and references therein.

\subsection{Lexicographic \grobner bases}
\label{sec:gb}

A \grobner basis of a polynomial ideal is a special set of generators of the ideal which is well structured and has good properties. The structures and properties of \grobner bases allow one to solve various computational problems with polynomial ideals, such as basic ideal operation, ideal membership test, and primary ideal decomposition. Introduced by Buchberger \cite{B1965A} in 1965 and having been developed for over half a century \cite{B85G,GTZ88g,w92c,FGLM93E,SY1996l,F1999A,F2002A,GVW2010,KSW2010a}, \grobner bases have become a truly powerful method that has applications everywhere polynomial ideals are involved.

Let $\fk$ be any field and $\fk[x_1, \ldots, x_n]$ be the ring of polynomials in the variables $x_1, \ldots, x_n$ with coefficients in $\fk$. In the sequel, we fix the variable order as $x_1 < \cdots < x_n$ unless otherwise specified. For the sake of simplicity, we write $\p{x}$ for $(x_1, \ldots, x_n)$, $\p{x}_i$ for $(x_1, \ldots, x_i)$, and $\kx$ for $\fk[x_1, \ldots, x_n]$.

A total ordering $<$ on all the terms in $\kx$ is called a \emph{term ordering} if it is a well ordering and for any terms $\p{u}, \p{v}$, and $\p{w}$ in $\kx$, $\p{u} > \p{v}$ implies $\p{u}\p{w} >
\p{v}\p{w}$. In this paper we are concerned mainly with the lex term ordering, with respect to which \grobner bases possess rich algebraic structures. For any two terms $\p{u} = \p{x}^{\p{\alpha}}$ and $\p{v} = \p{x}^{\p{\beta}}$ in $\kx$, we say that $\p{u} >_{\rm lex} \p{v}$ if the left rightmost nonzero entry in the vector $\p{\alpha} - \p{\beta}$ is positive.

Fix a term ordering $<$. The greatest term in a polynomial $F\in \kx$ with respect to $<$ is called the \emph{leading term} of $F$ and denoted by $\lt(F)$. As usual, $\bases{\{F_1, \ldots, F_s,\ldots\}}=\bases{F_1, \ldots, F_s,\ldots}$ denotes the ideal generated by the polynomials $F_1,\ldots,F_s,\ldots\in \kx$.

\begin{definition}
Let
$\mathfrak{I}\subseteq \kx$ be an ideal and $<$ be a term ordering. A finite
set $\{G_1, \ldots, G_s\}$ of polynomials in $\mathfrak{I}$ is called a \emph{\grobner
basis} of $\mathfrak{I}$ with respect to $<$ if $\bases{\lt(G_1), \ldots, \lt(G_s)} = \bases{\lt(\mathfrak{I})}$,
where $\lt(\mathfrak{I})$ denotes the set of leading terms of all the
polynomials in $\mathfrak{I}$.
\end{definition}

Let $\pset{G} = \{G_1, \ldots, G_s\}$ be a \grobner basis of an ideal $\mathfrak{I} \subseteq \kx$ with respect to a fixed term ordering $<$. For any polynomial $F\in \kx$, there is a unique polynomial $R\in \kx$ corresponding to $F$ such that $F-R \in \mathfrak{I}$ and no term of $R$ is divisible by any of $\lt(G_1), \ldots, \lt(G_s)$. The polynomial $R$ is called the \emph{normal form} of $F$ with respect to $\pset{G}$ (denoted by $\nform(F, \pset{G})$), and $F$ is said to be \emph{B-reduced} with respect to $\pset{G}$ if $F = R$. The \grobner basis $\pset{G}$ itself is said to be \emph{reduced} if the coefficient of each $G_i$ in $\lt(G_i)$ is $1$ and no term of $\pset{G}_i$ lies in $\bases{\{\lt(G)|\, G\in \pset{G}, G \neq G_i\}}$ for all $i=1, \ldots, s$. The reduced \grobner basis of $\mathfrak{I}$ with respect to a fixed term ordering is unique.

\begin{example}\label{ex:gb}
 Consider $\pset{P}=\{x_1x_2, x_2x_3+x_1^2x_2, x_3^2, x_1x_4^2, (x_1^2x_3+1)x_4+x_3x_2x_1\}\subseteq \fk[x_1, x_2, x_3, x_4]$. The polynomial set $\pset{G} = \{x_1x_2,x_2x_3,x_3^2,x_4\}$ is a \grobner basis of the ideal $\bases{\pset{P}}$ with respect to the lex ordering on $x_1<x_2<x_3<x_4$. One can easily check that $\pset{G}$ is also the reduced lex \grobner basis of $\pset{P}$.
\end{example}

\begin{remarks}
   The term \emph{B-reduced} is an abbreviation of \emph{Buchberger-reduced} for a polynomial modulo a \grobner basis. We use the prefix B to distinguish this term from the term \emph{R-reduced} (short for Ritt-reduced, defined below) for a polynomial modulo a triangular set.
\end{remarks}

\subsection{Normal triangular sets}
\label{sec:tri}

Now let $F$ be a polynomial in $\kx\setminus\fk$. With respect to the variable order, the greatest variable which actually appears in $F$ is called the \emph{leading variable} of $F$ and denoted by $\lv(F)$. Let $\lv(F) = x_i$; then $F$ can be written as $F = Ix_i^k + R$, with $I\in \fk[\p{x}_{i-1}]$, $R
\in \fk[\p{x}_i]$, and $\deg(R, x_i) < k=\deg(F, x_i)$. The polynomial
$I$ is called the \emph{initial} of $F$, denoted by $\ini(F)$. For any polynomial set $\pset{F}\subseteq \kx$, $\ini(\pset{F})$ denotes $\{\ini(F)\mid\, F\in \pset{F}\}$.

\begin{definition}
Any finite, nonempty, ordered set $[T_1, \ldots, T_r]$ of polynomials in $\kx\setminus\fk$ is called a \emph{triangular set} if $\lv(T_1) < \cdots < \lv(T_r)$ with respect to the variable order.
\end{definition}

The \emph{saturated ideal} of a triangular set $\pset{T} = [T_1, \ldots, T_r]$ is defined as $\sat(\pset{T}) = \bases{\pset{T}}: J^{\infty}$, where $J=\ini(T_1)\cdots \ini(T_r)$. We write $\sat_i(\pset{T}) = \sat([T_1,\ldots, T_i])$ for $i=1, \ldots, r$. The variables in $\{x_1, \ldots, x_n\} \setminus \{\lv(T_1), \ldots, \lv(T_r)\}$ are called the \emph{parameters} of $\pset{T}$. A triangular set $\pset{T}$ is said to be \emph{zero-dimensional} if there is no parameter of $\pset{T}$, and \emph{positive-dimensional} otherwise.

\begin{definition}
   A triangular set $\pset{T}=[T_1, \ldots, T_r] \subseteq \kx$ is said to be \emph{normal} (or called a \emph{normal set})
if all $\ini(T_1), \ldots, \ini(T_r)$ involve only the parameters of $\pset{T}$.
\end{definition}

\begin{example}\label{ex:tri}
 From the \grobner basis $\pset{G} = \{x_1x_2, x_2x_3, x_3^2, x_4\}$ in Example \ref{ex:gb} one can extract two triangular sets $\pset{T}_1=[x_1x_2,x_2x_3,x_4]$ and $\pset{T}_2=[x_1x_2,x_3^2,x_4]$. Both of them are positive-dimensional, with $x_1$ as their parameter. One can easily see that $\pset{T}_1$ is not normal, but $\pset{T}_2$ is.
\end{example}

 Among the most commonly used triangular sets there are \emph{regular sets} \cite{w00c} or \emph{regular chains} \cite{a99t}. A triangular set $\pset{T}=[T_1, \ldots, T_r] \subseteq \kx$ is called a regular set or said to be regular if for every $i=2, \ldots, r$, $\ini(T_i)$ is neither zero nor a zero-divisor in $\kx/\sat_{i-1}(\pset{T})$. By definition any normal set is obviously regular. It is proved in \cite{a99t,W2001E} that a triangular set $\pset{T}$ is regular if and only if $\sat(\pset{T}) = \{F|\prem(F, \pset{T}) = 0\}$.

A nonzero polynomial $P\in \kx$ is said to be \emph{R-reduced} with respect to $Q \in \kx\setminus\fk$ if $\deg(P, \lv(Q)) < \deg(Q, \lv(Q))$; $P$ is \emph{R-reduced} with respect to a triangular set $\pset{T}=[T_1, \ldots, T_r] \subseteq \kx$ if $P$ is R-reduced with respect to all $T_i$ for $i=1, \ldots, r$. A triangular set $\pset{T}$ itself is said to be \emph{R-reduced} if $T_i$ is R-reduced with respect to $[T_1, \ldots, T_{i-1}]$ for all $i=2, \ldots, r$.

Denote by $\prem(P, Q)$ the \emph{pseudo-remainder} and by $\pquo(P, Q)$ the \emph{pseudo-quotient} of $P\in \kx$ with respect to $Q\in \kx\setminus\fk$ in $\lv(Q)$, and for any triangular set $\pset{T}=[T_1, \ldots, T_r] \subseteq \kx$ define
$$\prem(P, \pset{T}) = \prem(\cdots\prem(\prem(P, T_r), T_{r-1}),\ldots, T_1),$$
 called the \emph{pseudo-remainder} of $P$ with respect to $\pset{T}$. Clearly, $\prem(P, Q)$ and $\prem(P, \pset{T})$ are respectively R-reduced with respect to $Q$ and $\pset{T}$. Similarly, we can define the \emph{resultant} of $P$ with respect to $\pset{T}$ as
  $$\res(P, \pset{T}) = \res(\cdots\res(\res(P, T_r), T_{r-1}),\ldots, T_1),$$
where $\res(P, Q)$ denotes the \emph{resultant} of $P\in \kx$ and $Q\in \kx\setminus\fk$ with respect to $\lv(Q)$.

\subsection{W-characteristic sets}

 From the reduced lex \grobner basis $\pset{G}$ of any polynomial ideal $\bases{\pset{P}}\subseteq \kx$, one can extract the W-characteristic set $\pset{C}$ of $\bases{\pset{P}}$ defined below.
\begin{definition}[{{\cite[Def.~3.1]{W2016o}}}]\label{def:wchar}
   Let $\pset{G}$ be the reduced lex \grobner basis of an ideal generated by an arbitrary polynomial set $\pset{P}\subseteq \kx$, and denote by $\pset{G}^{(i)} = \{G\in \pset{G} |\, \lv(G) = x_i\}$. Then the set
$$\bigcup_{i=1}^{n}\{G \in \pset{G}^{(i)}|\, \forall G'\in \pset{G}\setminus \{G\}, G <_{\rm lex} G' \},$$
ordered according to $<_{\rm lex}$, is called the \emph{W-characteristic set} of $\bases{\pset{P}}$.
\end{definition}

The set in Definition \ref{def:wchar} is also called the W-characteristic set of the reduced lex \grobner basis $\pset{G}$ for the sake of simplicity. By definition any W-characteristic set is a triangular set.

\begin{example}\label{ex:w-char}
Clearly among the two triangular sets $\pset{T}_1$ and $\pset{T}_2$ extracted from the reduced lex \grobner basis $\pset{G}$ in Example~\ref{ex:tri}, $\pset{T}_1$ is the W-characteristic set of $\bases{\pset{P}}$, for $x_2x_3 <_{\rm lex} x_3^2$.
\end{example}

\begin{definition}
  Let $\pset{P}$ be any finite polynomial set in $\kx$. An R-reduced triangular set $\pset{C} \subseteq \kx$ is called a \emph{Ritt characteristic set} of the ideal $\bases{\pset{P}}$ if $\pset{C} \subseteq \bases{\pset{P}}$ and for any $P \in \bases{\pset{P}}$, $\prem(P, \pset{C}) = 0$.
\end{definition}

The W-characteristic set $\pset{C}$ of $\bases{\pset{P}}$ is a Ritt characteristic set of $\bases{\pset{P}}$ if $\pset{C}$ is R-reduced \cite[Thm.\,3.3]{W2016o}. Further relationships between Ritt characteristic sets and lex \grobner bases are established in \cite{W2016o} with the help of the
concept of W-characteristic sets.

For any polynomial $P \in \kx$ and polynomial set $\pset{P} \subseteq \kx$, we denote by $\zero(P)$ the set of zeros of $P$ in $\bar{\fk}$, the algebraic closure of $\fk$, and by $\zero(\pset{P})$ the set of common zeros of all the polynomials in $\pset{P}$ in $\bar{\fk}^n$. For any nonempty polynomial sets $\pset{P}$ and $\pset{Q}$ in $\kx$, we define $\zero(\pset{P} / \pset{Q}) = \zero(\pset{P}) \setminus \zero(\Pi_{Q\in \pset{Q}}Q)$.

\begin{proposition}[{{\cite[Prop.\ 3.1]{W2016o}}}]\label{prop:zero}
 Let $\pset{C}$ be the W-characteristic set of $\bases{\pset{P}} \subseteq \kx$. Then
\begin{enumerate}
\item[$(a)$] for any $P \in \bases{\pset{P}}$, $\prem(P, \pset{C}) = 0$;
\item[$(b)$] $\bases{\pset{C}} \subseteq \bases{\pset{P}} \subseteq \sat(\pset{C})$;
\item[$(c)$] $\zero(\pset{C} / \ini(\pset{C})) \subseteq \zero(\pset{P}) \subseteq \zero(\pset{C})$.
\end{enumerate}
\end{proposition}

The following theorem \cite[Thm.\ 3.9]{W2016o} exploits the pseudo-divisibility relationships between polynomials in W-characteristic sets (and thus between those in reduced lex \grobner bases) for polynomial ideals of arbitrary dimension, while other well-known structural properties of lex \grobner bases were established only for bivariate or zero-dimensional polynomial ideals. It is these relationships that enable us to adopt an effective splitting strategy for our algorithm of characteristic decomposition.

\begin{theorem}[{{\cite[Thm.\ 3.9]{W2016o}}}]\label{thm:divisible}
Let $\pset{C} = [C_1, \ldots, C_r]$ be the W-characteristic set of $\bases{\pset{P}}\subseteq \kx$. If $\pset{C}$ is not normal, then there exists an integer $k~(1\leq k \leq r)$ such that $[C_1, \ldots, C_k]$ is normal and $[C_1, \ldots, C_{k+1}]$ is not regular.

 Assume that the variables $x_1, \ldots, x_n$ are ordered such that the parameters of $\pset{C}$ are all smaller than the other variables and let $I_{k+1} = \ini(C_{k+1})$ and $l$ be the integer such that $\lv(I_{k+1})=\lv(C_l)$.
\begin{enumerate}
\item[$(a)$]\label{item:main-a} If $I_{k+1}$ is not R-reduced with respect to $C_l$, then
  \begin{equation*}
    \label{eq:thm}
    \begin{split}
      & \prem(I_{k+1}, [C_1, \ldots, C_l]) = 0, \\
      & \prem(C_{k+1}, [C_1, \ldots, C_k]) = 0.
    \end{split}
 \end{equation*}
\item[$(b)$] If $I_{k+1}$ is R-reduced with respect to $C_l$, then
$$\prem(C_l, [C_1, \ldots, C_{l-1}, I_{k+1}])=0$$
 and either
$\res(\ini(I_{k+1}), [C_1, \ldots, C_{l-1}])=0$ or
$$\prem(C_{k+1}, [C_1, \ldots, C_{l-1}, I_{k+1}, C_{l+1}, \ldots, C_k])=0.$$
\end{enumerate}
\end{theorem}

\begin{example}\label{ex:wCharThm}
The W-characteristic set $\pset{T}_1=[x_1x_2,x_2x_3,x_4]$ of $\bases{\pset{P}}$ in Example~\ref{ex:w-char} is not normal, and one can find that $[x_1x_2]$ is normal, but $[x_1x_2,x_2x_3]$ is not (furthermore, it is not regular). The initial $x_2$ of $x_2x_3$ is not R-reduced with respect to $x_1x_2$, and one can check that $\prem(x_2,[x_1x_2])=0$ and $\prem(x_2x_3,[x_1x_2])=0$, which accord with Theorem~\ref{thm:divisible}(a).
\end{example}

An obvious consequence of the first part of the theorem is that the W-characteristic set contained in the reduced lex \grobner basis of a polynomial ideal, if it is regular, must be normal. This implies that certain normalization mechanism is integrated into the algorithm of \grobner bases, so that triangular subsets of lex \grobner bases are normalized as much as possible.

The condition on the order of $x_1, \ldots, x_n$ for $\pset{C}$ in Theorem~\ref{thm:divisible} is needed, for otherwise the theorem does not necessarily hold, as shown by \cite[Ex.~3.1(b)]{W2016o}. In the latter case, one may change the variable order properly to make the condition satisfied, so as to obtain the pseudo-divisibility relations in Theorem~\ref{thm:divisible}. For polynomial ideals of dimension $0$, their W-characteristic sets do not involve any parameters and thus the condition is satisfied naturally. For the rest of the paper, we assume that the condition is also satisfied for the positive-dimensional case where the structures of lex \grobner bases are rather complicated.

\section{Decomposition into characteristic pairs}
\label{sec:prop}
In this section we discuss the decomposition of an arbitrary polynomial set into (strong) characteristic pairs with associated zero relations and prove some properties about the decomposition. A decomposition algorithm will be presented in Section \ref{sec:alg}.

\subsection{Characteristic pairs and strong characteristic pairs}
\label{sec:prop-normal}

\begin{definition}
A pair $(\pset{G}, \pset{C})$ with $\pset{G}, \pset{C} \subseteq \kx$ is called a \emph{characteristic pair} in $\kx$ if $\pset{G}$ is a reduced lex \grobner basis, $\pset{C}$ is the W-characteristic set of $\bases{\pset{G}}$, and $\pset{C}$ is normal.
\end{definition}

The following known results (Propositions \ref{prop:normal-zero}--\ref{prop:normal-proj}) concerning normal sets are recalled, with references or self-contained proofs, exhibiting some of the nice properties of characteristic pairs.

\begin{proposition}\label{prop:normal-zero}
  For any zero-dimensional normal set $\pset{N} \subseteq \kx$: $(a)$ $\sat(\pset{N}) = \bases{\pset{N}}$; $(b)$ $\pset{N}$ is the lex \grobner basis of $\bases{\pset{N}}$.
\end{proposition}

\begin{proof}
  As $\pset{N}$ is a zero-dimensional normal set, each $\ini(N)$ is a constant in $\fk$ for $N\in \pset{N}$. Then statement (a) follows directly from the definition of $\sat(\pset{N})$, and statement (b) can be derived easily by using \cite[Section~2.9,  Thm.~3 and Prop.~4]{CLO1997I}.
\end{proof}

\begin{proposition}[{{\cite[Prop.~2.2]{M2012d}}}]\label{prop:normal-pos}
Let $\pset{N}\subseteq \kx$ be any positive-dimensional normal set with parameters $\tilde{\p{x}} \subseteq \p{x}$. Then $\pset{N}$ is the lex \grobner basis of $\bases{\pset{N}}$ over $\fk(\tilde{\p{x}})$.
\end{proposition}

Let $\pset{T} = [T_1, \ldots, T_r] \subseteq \kx$ be an arbitrary triangular set with parameters $x_1, \ldots, x_d~(d+r=n)$. For each $i=0, \ldots, r-1$, write
\begin{equation*}
\pset{T}_{\leq i} = \pset{T}\cap \fk[x_1, \ldots, x_{d+i}] = [T_1, \ldots, T_i], \quad
\pset{I}_{\leq i} = \ini(\pset{T}) \cap \fk[x_1, \ldots, x_{d+i}].
\end{equation*}
$\pset{T}$ is said to have the \emph{projection property} if for any $i=0$, $\ldots, r-1$ and any $\bar{\p{x}}_i \in \zero(\pset{T}_{\leq i}/ \pset{I}_{\leq i})$, there exist $\bar{x}_{i+1}, \ldots, \bar{x}_{r} \in \bar{\fk}$ such that $(\bar{\p{x}}_i, \bar{x}_{i+1}, \ldots, \bar{x}_{r}) \in \zero(\pset{T} / \ini(\pset{T}))$. Here empty $\pset{T}_{\leq i}$ and $\pset{I}_{\leq i}$ are understood as $\{0\}$ and $\{1\}$ respectively.

\begin{proposition}\label{prop:normal-proj}
Any normal set $\pset{N}\subset\kx$ has the projection property.
\end{proposition}

\begin{proof}
Let $\pset{N} = [N_1, \ldots, N_r]$, and $I_i = \ini(N_i)$ for $i=1, \ldots, r$. Since $\pset{N}$ is a normal set, $\pset{I}_{\leq i} = \{I_1, \ldots, I_r\}$ for $i=0, \ldots, r$. Thus for any $i=0, \ldots, r-1$ and any $\bar{\p{x}}_i \in \zero(\pset{N}_{\leq i} / \pset{I}_{\leq i})$, $I_j(\bar{\p{x}}_i)\neq 0$ for all $j=i+1, \ldots, r$, so there exist $\bar{x}_{i+1}, \ldots, \bar{x}_{r} \in \bar{\fk}$ such that $(\bar{\p{x}}_i, \bar{x}_{i+1}, \ldots, \bar{x}_{r}) \in \zero(\pset{N} / \ini(\pset{N}))$.
\end{proof}

\begin{remarks}
In general regular sets do not have the projection property. Consider, for example, $\pset{T} = [x^2-u, xy+1] \subseteq \qnum[u, x,y]$, where $\qnum$ is the field of rational numbers and $u<x<y$. Then $u$ is the parameter of $\pset{T}$. Now $\pset{T}_{\leq 0} =\pset{I}_{\leq 0} = \emptyset$ and the parametric value $\bar{u}=0 \in \bar{\qnum} = \zero(\pset{T}_{\leq 0} / \pset{I}_{\leq 0})$, but $\zero(\pset{T}/\ini(\pset{T}) ) = \emptyset$ when $u=\bar{u}$.
\end{remarks}

\begin{proposition}\label{prop:satT=T}
Let $\pset{C}$ be the normal W-characteristic set of $\bases{\pset{P}} \subseteq \kx$. If $\sat(\pset{C}) = \bases{\pset{C}}$, then $\sat(\pset{C})=\bases{\pset{P}}$.
\end{proposition}

\begin{proof}
The proposition follows immediately from $\bases{\pset{C}} \subseteq \bases{\pset{P}} \subseteq \sat(\pset{C})$ (Proposition~\ref{prop:zero}(b)).
\end{proof}

The reverse direction of Proposition~\ref{prop:satT=T}, namely $\sat(\pset{C}) = \bases{\pset{P}}$ implies $\sat(\pset{C}) = \bases{\pset{C}}$, is not correct in general. For example, $\pset{G} = \{y^2, xz+y, yz, z^2\}\subseteq \fk[x, y, z]$ is a reduced lex \grobner basis with $x<y<z$: the normal W-characteristic set of $\bases{\pset{G}}$ is $\pset{C} = [y^2, xz+y]$, and one can check that $\bases{\pset{G}} = \sat(\pset{C})$, but $\bases{\pset{C}} \neq \sat(\pset{C})$.

 What is of special interest between $\pset{G}$ and $\pset{C}$ in a characteristic pair $(\pset{G}, \pset{C})$ is whether the equality $\bases{\pset{G}} = \sat(\pset{C})$ holds. This equality does hold when $\sat(\pset{C}) = \bases{\pset{C}}$ (according to Proposition~\ref{prop:satT=T}), but the condition $\sat(\pset{C}) = \bases{\pset{C}}$ does not necessarily hold as the above example shows. Moreover, it is computationally difficult to verify whether $\sat(\pset{T}) = \bases{\pset{T}}$ holds for a triangular set $\pset{T}$ \cite{A2015t,L2011w}.

\begin{definition}\label{def:strong}
  A characteristic pair $(\pset{G}, \pset{C})$ is said to be \emph{strong} if $\sat(\pset{C})=\bases{\pset{G}}$.
\end{definition}

\begin{definition}\label{def:characterizable}
A reduced lex \grobner basis $\pset{G}$ is said to be \emph{characterizable} if $\bases{\pset{G}}=\sat(\pset{C})$, where $\pset{C}$ is the W-characteristic set of $\pset{G}$.
\end{definition}

It is easy to see that every W-characteristic set is determined by a reduced lex \grobner basis, while a characterizable \grobner basis is also determined by its W-characteristic set. A strong characteristic pair thus furnishes a characterizable \grobner basis with a normal W-characteristic set. In what follows  we show that the W-characteristic set of any characterizable \grobner basis is normal, so that the characterizable \grobner basis and its W-characteristic set form a strong characteristic pair.

\begin{proposition}\label{prop:strong}
  The W-characteristic set of any characterizable \grobner basis is normal.
\end{proposition}

\begin{proof}
We prove the proposition by contradiction. Suppose that the W-characteristic $\pset{C} = [C_1, \ldots, C_r]$ of the characterizable \grobner basis $\pset{G}$  is abnormal. Then by Theorem \ref{thm:divisible} there exist two polynomials $C_{k+1}$ and $C_l~(l \leq k)$ in $\pset{C}$ with $\lv(C_l)=\lv(I_{k+1})$ such that either (a) $\prem(\ini(C_{k+1}), [C_1, \ldots, C_l])) = 0$, when $\ini(C_{k+1})$ is not R-reduced with respect to $C_l$; or (b) $\prem(C_l, [C_1, \ldots, C_{l-1},$ $\ini(C_{k+1})]) = 0$, when $\ini(C_{k+1})$ is R-reduced.

Let $I_i = \ini(C_i)$ for $i=1, \ldots, l$ and $I_{k+1} = \ini(C_{k+1})$.

For case (a), from the pseudo-remainder formula we know that $I_{k+1}\in \sat(\pset{C})$. Write $C_{k+1} = I_{k+1}\lv(C_{k+1})^d + R$, where $\deg(R, \lv(C_{k+1})) < d$. If $R=0$, then $\lv(C_{k+1})^d \in \sat(\pset{C})$, but $\lv(C_{k+1})^d <_{\rm lex} C_{k+1}$, which contradicts with the minimality of $\pset{G}$ as the reduced lex \grobner basis of $\sat(\pset{C})$; If $R\neq 0$, clearly $R \in \sat(\pset{C})$, but $R \not \in \bases{\pset{G}}$ for $R$ is B-reduced with respect to $\pset{G}$, which contradicts the equality $\bases{\pset{G}} = \sat(\pset{C})$. 

For case (b), it follows from the pseudo-remainder formula that there exist $i_1, \ldots, i_l\in \znum_{\geq 0}$ (the set of nonnegative integers) and $Q_1, \ldots, Q_l \in \kx$ such that
$$I_1^{i_1}\cdots I_{l-1}^{i_{l-1}} \ini(I_{k+1})^{i_l} C_l = Q_1C_1 + \cdots + Q_lI_{k+1};$$
clearly $Q_l \in \sat(\pset{C})$. Since $\deg(I_{k+1}, \lv(C_l)) < \deg(C_l, \lv(C_l))$ in this case and all $I_1, \ldots, I_{l-1}$ involve only the parameters, we have $\lv(Q_l)=\lv(C_l)$ but $\deg(Q_l, \lv(C_l)) < \deg(C_l, \lv(C_l))$, and thus $Q_l <_{\rm lex} C_l$. This contradicts with the minimality of $\pset{G}$.
\end{proof}

\begin{remarks}
 \grobner bases are good representations of polynomial ideals. Here it is shown that the W-characteristic set $\pset{C}$ of a characterizable \grobner basis $\pset{G}$ provides another representation of the same ideal $\bases{\pset{G}}$. The representation $\pset{C}$ is simpler than $\pset{G}$ because $\pset{C}$ is a subset of $\pset{G}$, whereas $\pset{G}$ can be computed from $\pset{C}$ if needed. In other words, characterizable \grobner bases are those special \grobner bases whose W-characteristic sets can characterize or represent the ideals they generate.
\end{remarks}

\subsection{Characteristic decomposition and its properties}
\label{sec:dec-normal}

  Let $\pset{F}$ be a finite, nonempty set of polynomials in $\kx$.
  We call a finite set $\{(\pset{G}_1,\pset{C}_1), \ldots, (\pset{G}_{t}, \pset{C}_t)\}$ of characteristic pairs in $\kx$ a \emph{characteristic decomposition} of $\pset{F}$ if the following zero relations hold:
 \begin{equation} \label{eq:NormalDec}
\zero(\pset{F}) = \bigcup_{i=1}^t\zero(\pset{G}_i) = \bigcup_{i=1}^t\zero(\pset{C}_i / \ini(\pset{C}_i)) = \bigcup_{i=1}^t\zero(\sat(\pset{C}_i)).
 \end{equation}

\begin{theorem}\label{thm:main}
  From any finite, nonempty polynomial set $\pset{F}\subseteq \kx$,
  one can compute in a finite number of steps a characteristic decomposition of $\pset{F}$.
\end{theorem}

The above theorem will be proved by giving a concrete algorithm (Algorithm 1 in Section \ref{sec:alg-des}) with correctness and termination proof (in Section \ref{sec:correct}). In what follows, we focus our attention on the properties of characteristic decomposition.

\begin{remarks}
 From any characteristic decomposition of a polynomial set $\pset{F}$, one can extract a normal decomposition $\{\pset{C}_1, \ldots, \pset{C}_t\}$ of $\pset{F}$, with each $\pset{C}_i$ a normal set for $i=1, \ldots, t$ and $\zero(\pset{F}) = \bigcup_{i=1}^t\zero(\pset{C}_i / \ini(\pset{C}_i))$. The projection property of normal sets (see Proposition~\ref{prop:normal-proj}) allows us to write down the conditions on the parameters for a normal set to have zeros for the variables, which makes normal decomposition an appropriate approach for parametric polynomial system solving \cite{c07c,l91n,GC1992s}. The obtained parametric conditions are not necessarily disjoint and thus do not necessarily lead to a partition of the parameter space. However, since the conditions derived from normal sets are expressed by means of initials which involve only the parameters, it is easier to compute comprehensive triangular decompositions \cite{c07c} via normal decomposition than via regular decomposition.
\end{remarks}

A proper ideal in $\kx$ is said to be \emph{purely equidimensional} if its associated prime ideals are all of the same height. According to \cite[Prop.~4.1.3]{A1999e}, $\sat(\pset{T})$ is purely equidimensional for any regular set $\pset{T} \subseteq \kx$. Thus it follows from Proposition~\ref{prop:satT=T} that for any characteristic pair $(\pset{G}, \pset{C})$ in the characteristic decomposition of $\pset{F}$, $\bases{\pset{G}}$ is purely equidimensional if $\sat(\pset{C}) = \bases{\pset{C}}$ is verified. More precisely, we have the following.

\begin{proposition}\label{prop:equi}
  Let $\Psi$ be a characteristic decomposition of $\pset{F}\subseteq \kx$ and assume that $\sat(\pset{C}) = \bases{\pset{C}}$ for every characteristic pair $(\pset{G}, \pset{C}) \in \Psi$. Then
$\sqrt{\bases{\pset{F}}} = \bigcap \nolimits _{(\pset{G}, \pset{C})\in \Psi} \sqrt{\bases{\pset{G}}}$
and each $\bases{\pset{G}}$ is purely equidimensional.
\end{proposition}

In fact, the ideal $\bases{\pset{G}}$ in Proposition~\ref{prop:equi} is also \emph{strongly} equidimensional according to \cite[Thm.~4.1.4]{A1999e}. The following theorem shows how a Ritt characteristic set of a polynomial ideal can be constructed from the W-characteristic set (when it is normal) of the ideal.

\begin{theorem}\label{thm:reduced}
Let $\pset{C} \!=\! [C_1, \ldots, C_r]$ be the W-characteristic set of $\bases{\pset{P}} \subseteq \kx$ and
\begin{equation}
  \label{eq:cstar}
\pset{C}^* \!\!=\!\! [C_1, \prem(C_2, [C_1]), \ldots,\! \prem(C_r, [C_1, \!\ldots\!, C_{r-1}])].
\end{equation}
If $\pset{C}$ is normal, then the following statements hold:
\begin{enumerate}
\item[$(a)$] $\pset{C}^*$ is a normal set;
\item[$(b)$] $\pset{C}^*$ is a Ritt characteristic set of $\bases{\pset{P}}$;
\item[$(c)$] $\zero(\pset{C}^* / \ini(\pset{C}^*)) = \zero(\pset{C} / \ini(\pset{C}))$.
\end{enumerate}
\end{theorem}

\begin{proof}
Let $\pset{C}^* = [C_1^*, \ldots, C_r^*]$, $I_i = \ini(C_i)$, and $I_i^* = \ini(C_i^*)$ for $i=1, \ldots, r$.

(a--b) According to \cite[Thm.~3.4]{W2016o}, $\pset{C}^*$ is a regular set and $\pset{C}^*$ is a Ritt characteristic set of $\bases{\pset{P}}$; hence it suffices to prove that $\pset{C}^*$ is normal. Since $C^*_i = \prem(C_i, [C_1, \ldots, C_{i-1}])$ for any $i=1, \ldots, r$, there exist $q_1, \ldots, q_{i-1}\in \znum_{\geq 0}$ and $Q_1, \ldots$, $Q_{i-1}\in \kx$ such that
  \begin{equation}
    \label{eq:psuedo}
I_1^{q_1}\cdots I_{i-1}^{q_{i-1}}C_i = Q_1 C_1 + \cdots + Q_{i-1}C_{i-1} + C_i^*.
  \end{equation}
Since $\ini(C_i)$ does not involve any of $\lv(C_1), \ldots, \lv(C_{i-1})$, we have $\ini(C_i^*) = I_1^{q_1}\cdots I_{i-1}^{q_{i-1}}I_i$ for $i=1, \ldots, r$; thus $\pset{C}^*$ is normal.

(c) On one hand, for any $\bar{\p{x}} \in \zero(\pset{C} / \ini(\pset{C}))$,  $C_i(\bar{\p{x}}) = 0$ and $I_i(\bar{\p{x}}) \neq 0$ for $i=1, \ldots, r$. From \eqref{eq:psuedo} we know that $C_i^*(\bar{\p{x}}) = 0$ and $I_i^*(\bar{\p{x}}) \neq 0$ for $i=1, \ldots, r$; or equivalently $\bar{\p{x}} \in \zero(\pset{C}^* / \ini(\pset{C}^*))$. On the other hand, for any $\hat{\p{x}} \in \zero(\pset{C}^*/ \ini(\pset{C}^*))$, $C_i^*(\hat{\p{x}}) = 0$ and $I_i^{*}(\hat{\p{x}}) \neq 0$ for $i=1, \ldots, r$. Clearly $C_1 = C_1^*$, $C_1(\hat{\p{x}}) = 0$ and $I_1(\hat{\p{x}}) \neq 0$. Suppose now that $C_i(\hat{\p{x}}) = 0$ and $I_i(\hat{\p{x}})\neq 0$ hold for $i=2, \ldots, k-1$. Then by \eqref{eq:psuedo} and $I_k^* = I_1^{q_1}\cdots I_{k-1}^{q_{k-1}}I_k$ we have $C_k(\hat{\p{x}}) = 0$ and $I_k(\hat{\p{x}})\neq 0$. By induction, $\hat{\p{x}} \in \zero(\pset{C}/ \ini(\pset{C}))$.
\end{proof}

\begin{corollary}\label{cor:RittDec}
  Let $\Psi$ be a characteristic decomposition of $\pset{F}\subseteq \kx$  and $\pset{C}^*$ be computed from $\pset{C}$ according to \eqref{eq:cstar} for each characteristic pair $(\pset{G}, \pset{C}) \in \Psi$. Then
$$\zero(\pset{F}) = \bigcup \nolimits _{(\pset{G}, \pset{C})\in \Psi} \zero(\pset{C}^* / \ini(\pset{C}^*))=\bigcup \nolimits _{(\pset{G}, \pset{C})\in \Psi} \zero(\sat(\pset{C}^*)),$$
and $\pset{C}^*$\! is the Ritt characteristic set of $\bases{\pset{G}}$ for each $(\pset{G},\pset{C})\! \in\! \Psi$. \end{corollary}

\begin{corollary}\label{prop:wchar-gb}
Let $\pset{P}$, $\pset{C}$, and $\pset{C}^*$ be as in Theorem~\ref{thm:reduced} and $\tilde{\p{x}}\subseteq \p{x}$ be the parameters of $\pset{C}$. Then both $\pset{C}$ and $\pset{C}^*$ are lex \grobner bases of $\bases{\pset{P}}$ over $\fk(\tilde{\p{x}})$. Furthermore, let $\pset{C}^* = [C_1^*, \ldots, C_r^*]$. Then $[C_1^* \ini(C_1^*)^{-1}, \ldots, C_r^*\ini(C_r^*)^{-1}]$ is the reduced lex \grobner basis of $\bases{\pset{P}}$ over $\fk(\tilde{\p{x}})$.
\end{corollary}

\begin{proof}
By Proposition~\ref{prop:normal-pos}, $\pset{C}$ is the lex \grobner basis of $\bases{\pset{C}}$ over $\fk(\tilde{\p{x}})$. As $\sat(\pset{C}) = \bases{\pset{C}}$ over $\fk(\tilde{\p{x}})$, from Proposition~\ref{prop:satT=T} we know that $\pset{C}$ is the lex \grobner basis of $\bases{\pset{P}}$ over $\fk(\tilde{\p{x}})$. Moreover, since each $\ini(C_i^*)$ is a constant in $\fk(\tilde{\p{x}})$ for $i=1, \ldots, r$ and $\deg(C_j^*, \lv(C_i^*)) < \deg(C_i^*, \lv(C_i^*))$ ($1\leq i, j \leq r$, $j\neq i$), $[C_1^* \ini(C_1^*)^{-1}, \ldots, C_r^*\ini(C_r^*)^{-1}]$ is the reduced lex \grobner basis of $\bases{\pset{P}}$ over $\fk(\tilde{\p{x}})$.
\end{proof}

Note that the R-reduced normal set $\pset{C}^*$ computed in \eqref{eq:cstar} is not necessarily a subset of $\pset{G}$ (while the W-characteristic set $\pset{C}$ is). Corollary~\ref{prop:wchar-gb} does not hold when the W-characteristic set $\pset{C}$ is abnormal. In fact, the problem of constructing a Ritt characteristic set of an ideal $\ideal{I}$ from the reduced lex \grobner basis of $\ideal{I}$ was already studied by Aubry, Moreno Maza, and Lazard in their influential paper \cite{a99t} of 1999. It is the first author of this paper who pointed out with an example in \cite{W2016o} that the relevant results including Theorem 3.1, Proposition 3.4, and Theorem 3.2 in Section 3 of \cite{a99t} are flawed. He showed that the construction is simple and straightforward in the normal case and made the problem of effective construction open again for the abnormal case. The incorrectness of the construction process in \cite{a99t} is caused essentially by the non-closeness of addition on polynomials having pseudo-remainder $0$ with respect to an irregular triangular set $\pset{T}$, that is, $\prem(P, \pset{T}) = \prem(Q, \pset{T}) = 0$ does not necessarily imply $\prem(P+Q, \pset{T}) = 0$ for arbitrary polynomials $P$ and $Q$. The non-closeness seems to be a major obstacle for constructing abnormal Ritt characteristic sets of polynomial ideals.

\subsection{Strong characteristic decomposition}
\label{sec:strongCharDec}

Now we show that for any characteristic decomposition $\Psi = \{(\pset{G}_1, \pset{C}_1), \ldots, (\pset{G}_t, \pset{C}_t)\}$ of $\pset{F} \subseteq \kx$, one can explicitly transform $\Psi$ into a strong characteristic decomposition $\bar{\Psi} = \{(\bar{\pset{G}}_1, \bar{\pset{C}}_1), \ldots, (\bar{\pset{G}}_t, \bar{\pset{C}}_t)\}$, where each $(\bar{\pset{G}}_i, \bar{\pset{C}}_i)$ is a strong characteristic pair for $i=1, \ldots, t$.

\begin{lemma}[{{\cite[Lem.~2.4]{W2016o}}}]\label{lem:prem}
 Let $\pset{T}=[T_1, \ldots, T_r] \subseteq \kx$ be a regular set with $\lv(T_r)< x_n$, and $P = P_dx_m^d + \cdots + P_1 x_m + P_0 \in \kx$ be a polynomial with $\lv(P) = x_m >\lv(T_r)$ and $\deg(P,x_m)=d$. Then $\prem(P, \pset{T}) = 0$ if and only if $\prem(P_i, \pset{T}) = 0$ for all $i=0, 1, \ldots, d$.
\end{lemma}

\begin{lemma}[{{\cite[Lem.~6.2.6]{W2001E}}}]\label{lem:doubleSat}
  Let $\pset{T}$ be a regular set in $\kx$. Then for any $F\in \kx$, if $\res(F, \pset{T}) \neq 0$, then $\sat(\pset{T}):F^{\infty} = \sat(\pset{T})$.
\end{lemma}

\begin{theorem}\label{prop:equiv}
  Let $(\pset{G}, \pset{C})$ be a characteristic pair, $\bar{\pset{G}}$ be the reduced lex \grobner basis of $\sat(\pset{C})$, and $\bar{\pset{C}}$ be the W-characteristic set of $\bases{\bar{\pset{G}}}$. Then the following statements hold:
  \begin{enumerate}
  \item[$(a)$] the parameters of $\bar{\pset{C}}$ coincide with those of $\pset{C}$;
  \item[$(b)$] $\bar{\pset{C}}$ is a normal set, $\sat(\bar{\pset{C}}) = \sat(\pset{C})$, and thus $(\bar{\pset{G}}, \bar{\pset{C}})$ is a strong characteristic pair.
    \end{enumerate}

\end{theorem}

\begin{proof}
(a) Let $\lv(\pset{G}) = \{\lv(G):\, G\in \pset{G}\}$ and $\lv(\bar{\pset{G}}) = \{\lv(\bar{G}):\, \bar{G}\in \bar{\pset{G}}\}$. It suffices to prove that $\lv(\pset{G}) = \lv(\bar{\pset{G}})$.

$(\lv(\pset{G}) \supseteq \lv(\bar{\pset{G}}))$ Suppose that $x_k \in \lv(\bar{\pset{G}})$, but $x_k \not \in \lv(\pset{G})$. Then there exists a $\bar{G}\in \bar{\pset{G}}$ with $\lv(\bar{G})=x_k$. Since $\bar{G}\in \bases{\bar{\pset{G}}} = \sat(\pset{C})$, $\prem(\bar{G}, \pset{C}) = 0$. Write $\bar{G} = \bar{H}_p x_k^p + \cdots + \bar{H}_0$, where $p=\deg(\bar{G})$ and $\bar{H}_i\in \fk[x_1, \ldots, x_{k-1}]$ for $i=0, \ldots, p$. Since $x_k\not \in \lv(\pset{G})$ and $\prem(\bar{G}, \pset{C}) = 0$, from Lemma \ref{lem:prem} we know that for $i=0,\ldots, p$, $\prem(\bar{H}_i, \pset{C}) = 0$ and thus $\bar{H}_i \in \sat(\pset{C}) = \bases{\bar{\pset{G}}}$. This implies that $\nform(\bar{H}_i, \bar{\pset{G}}\setminus \{\bar{G}\})=0$ for $i=0, \ldots, p$, and thus $\nform(\bar{G}, \bar{\pset{G}}\setminus \{\bar{G}\})=0$, which contradicts with $\bar{\pset{G}}$ being the reduced lex \grobner basis.

$(\lv(\pset{G}) \subseteq \lv(\bar{\pset{G}}))$ Suppose that $x_l \in \lv(\pset{G})$, but $x_l \not \in \lv(\bar{\pset{G}})$. Then there exists a $C \in \pset{C} \subseteq \pset{G}$ with $\lv(C) = x_l$, and thus $C \in \sat(\pset{C}) = \bases{\bar{\pset{G}}}$; it follows that $\nform(C, \bar{\pset{G}}) = 0$. Write $C = \ini(C)x_l^d + R$, where $d = \deg(C, x_l)$ and $R\in \fk[x_1, \ldots, x_l]$. Then from $\nform(C, \bar{\pset{G}})=0$ and $x_l \not \in \lv(\bar{\pset{G}})$, we know that $\nform(\ini(C), \bar{\pset{G}}) = 0$, and thus $\ini(C) \in \bases{\bar{\pset{G}}} = \sat(\pset{C})$. It follows that $\prem(\ini(C), \pset{C}) = 0$; this contradicts with the fact that $C\in \pset{C}$ and $\pset{C}$ is a normal set.

(b) Let $\pset{C}=[C_1, \ldots, C_r]$. By (a) we can assume that $\bar{\pset{C}} = [\bar{C}_1, \ldots, \bar{C}_r]$ with $\lv(\bar{C}_i) = \lv(C_i)$ for $i=1, \ldots, r$. Since $\bar{C}_i \in \bases{\bar{\pset{G}}} = \sat(\pset{C})$ and $\pset{C}$ is regular, $\prem(\bar{C}_i, \pset{C}) = 0$; since $C_i \in \sat(\pset{C}) = \bases{\bar{\pset{G}}}$ and $\bar{\pset{C}}$ is the W-characteristic set of $\bar{\pset{G}}$, $\prem(C_i, \bar{\pset{C}})=0$. This leads to $\deg(\bar{C}_i, \lv(\bar{C}_i)) = \deg(C_i, \lv(C_i))$. 

We first prove that $\bar{\pset{C}}$ is a normal set, namely $\bar{I}_i := \ini(\bar{C}_i)$ only involves the parameters for each $i=1, \ldots, r$. If, otherwise, some $\bar{I}_i$ involves the variables in $\lv(\bar{\pset{G}})$, then, under the assumption that all the parameters are ordered smaller than the variables in $\lv(\bar{\pset{G}})$, we have
$C_i <_{lex} \bar{C_i}$, for $\deg(\bar{C}_i, \lv(\bar{C}_i)) = \deg(C_i, \lv(C_i))$. But $C_i \in \bases{\bar{\pset{G}}}$, and this contradicts with the minimality of $\bar{C_i}$ as an element in $\bar{\pset{G}}$. 

Next we show the equality $\sat(\pset{\bar{C}}) = \sat(\pset{C})$. Since $\bases{\bar{\pset{G}}} = \sat(\pset{C})$, it suffices to show that $\bases{\bar{\pset{G}}} = \sat(\bar{\pset{C}})$. On one hand, from Proposition~\ref{prop:zero}(b) we know that $\bases{\bar{\pset{G}}} \subseteq \sat(\bar{\pset{C}})$. On the other hand, let $\bar{I} = \prod_{\bar{C}\in \bar{\pset{C}}} \ini(\bar{C})$. Since $\bar{I}$ only involves the parameters of $\pset{C}$, $\res(\bar{I}, \pset{C}) = \bar{I} \neq 0$. By Lemma \ref{lem:doubleSat} we have $\sat(\pset{C}):\bar{I}^{\infty} = \bases{\bar{\pset{G}}}$. With the inclusion $\bar{\pset{C}}\subseteq \bar{\pset{G}}$, the following relation
$$\sat(\bar{\pset{C}}) = \bases{\bar{\pset{C}}}:\bar{I}^{\infty} \subseteq \bases{\bar{\pset{G}}}:\bar{I}^{\infty} = \sat(\pset{C}):\bar{I}^{\infty} = \bases{\bar{\pset{G}}}$$
holds.
\end{proof}

\begin{example}
The pair
$$(\pset{G}, \pset{C}) = (\{y^2, x^2z+xy, yz+xz+y\}, [y^2,x^2z+xy])$$
is a characteristic pair in $\qnum[x, y, z]$ with $x<y<z$. The reduced lex \grobner basis $\bar{\pset{G}}$ of $\sat(\pset{C})$ is $\{y^2, xz+y, yz, z^2\}$, so the W-characteristic set $\bar{\pset{C}}$ of $\bar{\pset{G}}$ is $[y^2, xz+y]$. One can check that the parameter of both $\pset{C}$ and $\bar{\pset{C}}$ is $x$, $\bar{\pset{C}}$ is normal, and $\sat(\pset{C}) = \sat(\bar{\pset{C}})$.
\end{example}

\begin{lemma}\label{lem:zeroStrong}
  Let $\pset{C}=[C_1, \ldots, C_r]$ and $\bar{\pset{C}} = [\bar{C}_1, \ldots, \bar{C}_r]$ be as in Theorem~\ref{prop:equiv}. Then $\ini(\bar{C}_i) | \ini(C_i)$ for $i=1, \ldots, r$.
\end{lemma}

\begin{proof}
From the proof of Theorem~\ref{prop:equiv} we know that $\prem(\bar{C}_i, \pset{C}) = 0$, $\prem(C_i, \bar{\pset{C}})$ = 0, and $\deg(\bar{C}_i, \lv(\bar{C}_i)) = \deg(C_i, \lv(C_i))$ for $i=1, \ldots, r$. Then
$$\ini(C_i) \bar{C}_i = \ini(\bar{C}_i) C_i \mod \sat(\pset{C}_{\leq i-1}),$$
where $\pset{C}_{\leq i-1} = [C_1, \ldots, C_{i-1}]$. Therefore, $\ini(\bar{C}_i) \,|\, \ini(C_i) \bar{C}_i$ modulo $\sat(\pset{C}_{\leq i-1})$.

If $\ini(\bar{C}_i) \,|\, \ini(C_i)$ modulo $\sat(\pset{C}_{\leq i-1})$, then clearly $\ini(\bar{C}_i) \,|\, \ini(C_i)$, for both $\ini(\bar{C}_i)$ and $\ini(C_i)$ involve only the parameters. Otherwise, there exists an $\bar{I}\not \in \fk$ such that $\bar{I} \,|\, \ini(\bar{C}_i)$ and $\bar{I} \,|\, \bar{C}_i$ modulo $\sat(\pset{C}_{\leq i-1})$. Then by $(\bar{C}_i / \bar{I}) \bar{I} = \bar{C}_i \in \bases{\bar{\pset{C}}}$, we have $\bar{C}_i / \bar{I} \in \sat(\bar{\pset{C}}) = \bases{\bar{\pset{G}}}$, but $\bar{C}_i / \bar{I} <_{\rm lex} \bar{C}_i$, which contradicts with the minimality of $\bar{C}_i$ as a polynomial in the W-characteristic set.
\end{proof}

\begin{theorem}\label{thm:zeroStrong}
  Let $\Psi = \{(\pset{G}_1, \pset{C}_1), \ldots, (\pset{G}_t, \pset{C}_t)\}$ be a characteristic decomposition of $\pset{F}\subseteq \kx$. For each $(\pset{G}_i, \pset{C}_i)\in \Psi$,
let $(\bar{\pset{G}}_i, \bar{\pset{C}}_i)$ be the corresponding strong characteristic pair as constructed in Theorem \ref{prop:equiv}, $i=1, \ldots, t$. Then
$$\zero(\pset{F}) = \bigcup_{i=1}^t\zero(\bar{\pset{G}_i}) = \bigcup_{i=1}^t \zero(\bar{\pset{C}_i} / \ini(\bar{\pset{C}_i}))=\bigcup_{i=1}^t \zero(\sat(\bar{\pset{C}_i})).$$
\end{theorem}

\begin{proof}
 Note that the zero relation \eqref{eq:NormalDec} holds for the characteristic decomposition $\Psi$, which may be computed according to Theorem \ref{thm:main}.
 The first and the third equality above can be proved by using the equalities $\sat(\pset{C}_i) = \sat(\bar{\pset{C}_i})$ in Theorem~\ref{prop:equiv}(b) and $\bases{\bar{\pset{G}_i}} = \sat(\pset{C}_i)$:
$$\zero(\pset{F}) = \bigcup_{i=1}^t\zero(\sat(\pset{C}_i)) = \bigcup_{i=1}^t\zero(\sat(\bar{\pset{C}_i})) = \bigcup_{i=1}^t\zero(\bar{\pset{G}_i}).$$

Now we prove that $\zero(\pset{F}) =  \bigcup_{i=1}^t \zero(\bar{\pset{C}_i} / \ini(\bar{\pset{C}_i}))$. Since $\zero(\bar{\pset{C}_i} / \ini(\bar{\pset{C}_i})) \subseteq \zero(\sat(\bar{\pset{C}_i}))$, we have
$$\bigcup_{i=1}^t \zero(\bar{\pset{C}_i} / \ini(\bar{\pset{C}_i})) \subseteq \bigcup_{i=1}^t \zero(\sat(\bar{\pset{C}_i})) = \zero(\pset{F}).$$
To show the other inclusion, observe first that
\begin{equation}
  \label{eq:zeroStrong}
\zero(\pset{C}_i / \ini(\pset{C}_i)) \subseteq \zero(\sat(\pset{C}_i)) = \zero(\bar{\pset{G}_i}) \subseteq \zero(\bar{\pset{C}_i}).
\end{equation}
for $i=1, \ldots, t$.
  Write $\pset{C}_i = [C_{i1}, \ldots, C_{ir_i}]$ and $\bar{\pset{C}}_i = [\bar{C}_{i1}, \ldots, \bar{C}_{ir_i}]$. Then for any $\bar{\p{x}} \in \zero(\pset{C}_i / \ini(\pset{C}_i))$, $\prod_{j=1}^{r_i} \ini(C_{ij})(\bar{\p{x}}) \neq 0$. By Lemma~\ref{lem:zeroStrong} we have $\ini(\bar{C}_{ij}) | \ini(C_{ij})$ for $j=1, \ldots, r_i$, and thus $\prod_{j=1}^{r_i}\ini(\bar{C}_{ij})(\bar{\p{x}}) \neq 0$. Combining this inequality with \eqref{eq:zeroStrong}, we have $\zero(\pset{C}_i / \ini(\pset{C}_i)) \subseteq \zero(\bar{\pset{C}_i} /  \ini(\bar{\pset{C}_i}))$, and thus $\zero(\pset{F}) = \bigcup_{i=1}^t \zero(\pset{C}_i / \ini(\pset{C}_i)) \subseteq \bigcup_{i=1}^t \zero(\bar{\pset{C}_i} / \ini(\bar{\pset{C}_i}))$. This completes the proof.
\end{proof}

\begin{remarks}
As shown by Theorems \ref{prop:equiv} and \ref{thm:zeroStrong}, any characteristic decomposition $\Psi$ of a polynomial set can be transformed into a strong characteristic decomposition of the polynomial set by computing the reduced lex \grobner basis of the saturated ideal of the W-characteristic set in each characteristic pair in $\Psi$. No splitting occurs in this transformation: the number of strong characteristic pairs produced by the transformation is the same as that of the characteristic pairs in $\Psi$.
\end{remarks}

\section{Algorithm for characteristic decomposition}
\label{sec:alg}

In this section we present an algorithm that computes a characteristic decomposition of any finite, nonempty set of nonzero polynomials.

\subsection{Algorithm description}
\label{sec:alg-des}
An overall strategy based on Theorem~\ref{thm:divisible} for characteristic decomposition is sketched in \cite[Sect.\,4]{W2016o}. Following this strategy, we detail the decomposition method below.

Let $\Phi$ be a set of polynomial sets, initialized as $\{\pset{F}\}$ with $\pset{F} \subseteq \kx$ being the input set, and $\Psi$ be a set of characteristic pairs already computed. Now we pick a polynomial set $\pset{P}\in\Phi$ and remove it from $\Phi$, compute the reduced lex \grobner basis $\pset{G}$ of the ideal $\bases{\pset{P}}$, and extract the W-characteristic set $\pset{C}=[C_1, \ldots, C_r]$ of $\bases{\pset{P}}$ from $\pset{G}$. Let $I_i = \ini(C_i)$ for $i=1, \ldots, r$.

\begin{itemize}
\item[1.] If $\pset{C}$ is normal, then from Proposition~\ref{prop:zero}(c) we know that
\begin{equation}
  \label{eq:zero}
\zero(\pset{C} / \ini(\pset{C})) \subseteq \zero(\pset{P}) \subseteq \zero(\pset{C}).
\end{equation}
In view of this zero relation, we put the characteristic pair $(\pset{G}, \pset{C})$ into $\Psi$ and adjoin the polynomial sets $\pset{G} \cup \{I_1\}, \ldots, \pset{G} \cup \{I_r\}$ to $\Phi$ for further processing.
\item[2.] If $\pset{C}$ is not normal, then by Theorem~\ref{thm:divisible} we have certain pseudo-divisibility relations between polynomials in $\pset{C}$ and can use them to split $\pset{G}$ as follows, where the integers $k$ and $l$ are as in Theorem~\ref{thm:divisible}.
  \begin{itemize}
  \item[2.1.] If $I_{k+1}$ is not R-reduced with respect to $C_l$, then the polynomial sets $\pset{G}\cup \{I_1\}, \ldots, \pset{G}\cup \{I_l\}, \pset{G} \cup \{I_{k+1}\}$ are adjoined to $\Phi$.
  \item[2.2.] If $I_{k+1}$ is R-reduced with respect to $C_l$, then let $Q = \pquo(C_l, I_{k+1})$ be the pseudo-quotient of $C_l$ with respect to $I_{k+1}$ and $\pset{C}_{l-1}=[C_1, \ldots, C_{l-1}]$.
    \begin{itemize}
    \item[2.2.1.] If $\prem(\ini(Q), \pset{C}_{l-1}) =0$, then the polynomial sets $\pset{G} \cup \{I_1\}, \ldots, \pset{G} \cup \{I_{l-1}\}$, $\pset{G} \cup \{\ini(I_{k+1})\}$ are adjoined to $\Phi$.
    \item[2.2.2.] If $\prem(\ini(Q), \pset{C}_{l-1}) \neq 0$, then the polynomial sets $\pset{G} \cup \{I_1\}, \ldots, \pset{G}\cup \{I_{l-1}\}$, $\pset{G} \cup \{\prem(Q, \pset{C}_{l-1})\}$, $\pset{G}\cup \{I_{k+1}\}$ are adjoined to $\Phi$.
    \end{itemize}
  \end{itemize}
\end{itemize}

After the splitting of $\pset{G}$, we continue picking another polynomial set $\pset{P}'$ (and meanwhile remove it) from $\Phi$, compute the reduced lex \grobner basis $\pset{G}'$ of $\bases{\pset{P}'}$, extract the W-characteristic set of $\bases{\pset{P}'}$ from $\pset{G}'$, and split $\pset{G}'$ when necessary. This process is repeated until $\Phi$ becomes empty.

The method for characteristic decomposition, whose main steps are outlined above, is described formally as Algorithm~1.

~\\
{\bf Algorithm 1 $\Psi = \algnor(\pset{F})$}. Given a finite, nonempty set $\pset{F}$ of nonzero polynomials in $\kx$, this algorithm computes a characteristic decomposition $\Psi$ of $\pset{F}$, or the empty set meaning that $\zero(\pset{F})=\emptyset$.

\begin{itemize}[noitemsep,topsep=0pt]
\item[C1.] Set $\Psi = \emptyset$ and $\Phi = \{\pset{F}\}$.
\item[C2.] Repeat the following steps until $\Phi = \emptyset$:
  \begin{itemize}[noitemsep,topsep=0pt]
  \item[C2.1.] Pick $\pset{P}\in \Phi$ and remove it from $\Phi$.
  \item[C2.2.] Compute the reduced lex \grobner basis $\pset{G}$ of $\bases{\pset{P}}$.
  \item[C2.3.] If $\pset{G} = \{1\} $, then go to C2; otherwise:
    \begin{itemize}[noitemsep,topsep=0pt]
    \item[C2.3.1.] Extract the W-characteristic set $\pset{C}$ from $\pset{G}$.
    \item[C2.3.2.] If $\pset{C}$ is normal (case 1), then reset $\Psi$ with $\Psi \cup \{(\pset{G},\pset{C})\}$, and $\Phi$ with
      \begin{equation}
        \label{eq:case1}
\Phi \cup \{\pset{G} \cup \{\ini(C)\}\mid \ini(C)\not\in \fk, C\in \pset{C}\},
      \end{equation}
and go to C2.
    \item[C2.3.3.] Pick the smallest polynomial $C$ in $\pset{C}$ such that $[T\in\pset{C}\mid \lv(T)\leq\lv(C)]$ as a triangular set is abnormal.
    \item[C2.3.4.] Let $I = \ini(C)$ and $y = \lv(I)$.
    \item[C2.3.5.] Pick the polynomial $C^*$ in $\pset{C}$ such that $\lv(C^*) = y$.
    \item[C2.3.6.] If $I$ is not R-reduced with respect to $C^*$ (case 2.1), then reset $\Phi$ with
      \begin{equation}
        \label{eq:case2}
        \Phi \cup \{\pset{G}\cup \{\ini(T)\}\mid\lv(T)\leq y, T\in\pset{C}\} \cup \{\pset{G}\cup \{I\}\},
      \end{equation}
and go to C2.
    \item[C2.3.7.] If $\prem(\ini(Q), [T\in\pset{C}\mid \lv(T)<y]) = 0$ (case 2.2.1), then reset $\Phi$ with
      \begin{equation}
        \label{eq:case3}
\Phi \cup \{\pset{G}\cup \{\ini(T)\}\mid\lv(T)<y, T\in\pset{C}\} \cup \{\pset{G}\cup \{\ini(I)\}\},
      \end{equation}
and go to C2.
    \item[C2.3.8.] Reset $\Phi$ with
      \begin{equation}
        \label{eq:case4}
        \begin{split}
          \Phi \,\cup\,\{ &\pset{G}\cup \{\ini(T)\mid \lv(T)< y, T\in\pset{C}\}\}\, \cup\\
\{\pset{G}\,\cup\, & \{\prem(Q, [T\in\pset{C}\mid\,\lv(T)<y])\}, \pset{G}\cup \{I\}\}
        \end{split}
      \end{equation}
(case 2.2.2) and go to C2.
  \end{itemize}
\end{itemize}
\item[C3.] Output $\Psi$.
\end{itemize}

\subsection{Correctness and termination}
\label{sec:correct}

\begin{theorem}
  Algorithm~1 terminates in a finite number of steps with correct output.
\end{theorem}

\begin{proof}
  ({\em Termination}) The process of splitting in Algorithm~1 can be viewed as building up a tree from its root as the input set $\pset{F}$. Every time a polynomial set $\pset{P}$ is picked from $\Phi$, splitting occurs according to one of the four cases, treated in steps C2.3.2, C2.3.6, C2.3.7, and C2.3.8 of Algorithm~1, as long as the reduced lex \grobner basis of $\bases{\pset{P}}$ is not $\{1\}$ (for otherwise $\pset{P}$ has no zero). Suppose that the split polynomial sets $\pset{G}_1, \ldots, \pset{G}_s$ are adjoined to $\Phi$. In the sense of building up the tree, this means that the child nodes of $\pset{P}$ are $\pset{G}_1, \ldots, \pset{G}_s$.

To prove the termination of Algorithm~1, we need to show that each path in the tree is of finite length. Thus by the Ascending Chain Condition (see, e.g., \cite[Chap.\ 2, Thm.\ 7]{CLO1997I}), it suffices to show that for all the four cases of splitting, each polynomial set $\pset{G}' = \pset{G} \cup \{H\}$ adjoined to $\Phi$ for some $H$ generates an ideal $\bases{\pset{G'}}$ that is strictly greater than $\bases{\pset{G}}$, or equivalently $H \not \in \bases{\pset{G}}$.

Let the W-characteristic set $\pset{C}$ in step C2.3.1 be written as $[C_1,\ldots,C_r]$ with $I_i=\ini(C_i)$ and $\pset{C}_{i}=[C_1,\ldots,C_i]$ for $1\leq i\leq r$. Then $C$, $I$, and $C^*$ in steps C2.3.3--C2.3.5 correspond to $C_{k+1}, I_{k+1}$, and $C_l$ respectively for some integers $l\leq k$ as stated in Theorem~\ref{thm:divisible}.

Let $H \not \in \fk$ be the initial of some $C\in \pset{C}$ as in \eqref{eq:case1} in step C2.3.2. We claim that $H$ is B-reduced with respect to $\pset{G}$; for otherwise $C$ will be reducible by some polynomial in $\pset{G}\setminus \{C\}$, which conflicts with the fact that $\pset{G}$ is the reduced lex \grobner basis of $\bases{\pset{P}}$. Therefore, $H \not \in \bases{\pset{G}}$ for all $C\in \pset{C}$ in \eqref{eq:case1}. With similar arguments, one can show that $H \not \in \bases{\pset{G}}$ as well if $H \not \in \fk$ is the initial of some polynomial $T\in \pset{C}$ as in \eqref{eq:case2}, \eqref{eq:case3}, and \eqref{eq:case4}.

To complete the proof of termination, it remains to show that $H=\prem(Q, \pset{C}_{l-1})\not \in \bases{\pset{G}}$, where $Q=\pquo(C_l, I_{k+1})$.  In step C2.3.8, the conditions $\lv(Q) >\lv(C_{l-1})$ and $\prem(\ini(Q),$ $\pset{C}_{l-1}) \!\neq\! 0$ hold. By Theorem~\ref{thm:divisible}, $\pset{C}_{l-1}$ is normal and thus regular; then by Lemma~\ref{lem:prem}, $\prem(Q, \pset{C}_{l-1}) \neq 0$ and
\begin{equation}
  \label{eq:split4}
\deg(\prem(Q, \pset{C}_{l-1}), \lv(C_i)) < \deg(C_i, \lv(C_i))
\end{equation}
for $i=1, \ldots, l-1$. Furthermore, as $Q = \pquo(C_l, I_{k+1})$ and $\lv(C_l)$ also appears in $I_{k+1}$ (by Theorem~\ref{thm:divisible}), we have $\deg(Q, \lv(C_l)) < \deg(C_l, \lv(C_l))$ and thus
\begin{equation}
  \label{eq:split4-2}
\deg(\prem(Q, \pset{C}_{l-1}), \lv(C_l)) < \deg(C_l, \lv(C_l)).
\end{equation}
Since $\pset{C}$ is the W-characteristic set of $\pset{G}$, the relations \eqref{eq:split4} and \eqref{eq:split4-2} imply that $\prem(Q, \pset{C}_{l-1})$ is B-reduced with respect to $\pset{G}$ and thus $\prem(Q, \pset{C}_{l-1}) \not \in \bases{\pset{G}}$.

({\em Correctness}) When $\zero(\pset{F}) = \emptyset$, $\pset{G} = \{1\}$ in step C2.3 and $\Psi = \emptyset$ is returned. Therefore, to prove the correctness of Algorithm~1, we need to show that when $\zero(\pset{F}) \neq \emptyset$, $\Psi$ is a characteristic decomposition of $\pset{F}$, namely all the pairs $(\pset{G}, \pset{C})\in \Psi$ are characteristic pairs and the zero relation \eqref{eq:NormalDec} holds.

It is clear that each $(\pset{G}, \pset{C})\in \Psi$ is a characteristic pair, for only in step C2.3.2 is the output set $\Psi$ adjoined with a new pair $(\pset{G}, \pset{C})$, where $\pset{G}$ is a reduced lex \grobner basis and $\pset{C}$ is its normal W-characteristic set. We first prove that $\zero(\pset{F}) = \bigcup_{(\pset{G}, \pset{C}) \in \Psi}\zero(\pset{C} / \ini(\pset{C}))$ by considering all the four cases of splitting. For this purpose, let $\pset{C} = [C_1, \ldots, C_r]$ be the W-characteristic set, $I_i = \ini(C_i)$ $(1 \leq i \leq r)$, and $L = \ini(I_{k+1})$.

Case 1 (step C2.3.2): the W-characteristic set $\pset{C}$ is normal. In this case, let $J = I_1\cdots I_r$. Then $\zero(\pset{P}) = (\zero(\pset{P})\setminus \zero(J)) \cup \zero(\pset{P} \cup \{J\})$. It follows from the zero relation \eqref{eq:zero} that $\zero(\pset{P})\setminus \zero(J) = \zero(\pset{C} / \ini(\pset{C}))$. Moreover, as $J = 0$ implies that $I_1=0$, or $I_2 = 0, \ldots$, or $I_r = 0$, $$\zero(\pset{P} \cup \{J\}) = \bigcup_{i=1}^r \zero(\pset{P} \cup \{I_i\}) = \bigcup_{i=1}^r\zero(\pset{G} \cup \{I_i\}).$$ Therefore,
\begin{equation}
  \label{eq:zero1}
\zero(\pset{P}) = \zero(\pset{C} / \ini(\pset{C})) \cup \bigcup \limits _{i=1}^r\zero(\pset{G} \cup \{I_i\}).
\end{equation}

Case 2.1 (step C2.3.6): $\pset{C}$ is abnormal and $\deg(I_{k+1}, \lv(I_{k+1})) \geq \deg(C_l, \lv(I_{k+1}))$. By Theorem~\ref{thm:divisible}(a) we have $\prem(I_{k+1}, \pset{C}_l) = 0$, and thus there exist $Q_1, \ldots, Q_l \in \kx$ and $q_1, \ldots, q_l\in \znum_{\geq 0}$ such that
$I_1^{q_1}\cdots I_l^{q_l} I_{k+1}$ $= Q_1 C_1 + \cdots + Q_l C_l$. This means that $I_1^{q_1}\cdots I_l^{q_l} I_{k+1} \in \bases{C_1, \ldots, C_l} \subseteq \bases{\pset{P}} = \bases{\pset{G}}$, so
\begin{equation}\label{eq:zero2}
  \zero(\pset{P}) = \zero(\pset{G}) =\zero(\pset{G} \cup \{I_1^{q_1}\cdots I_l^{q_l} I_{k+1}\}) = \bigcup \limits _{i=1}^l\zero(\pset{G} \cup \{I_i\}) \cup \zero(\pset{G} \cup \{I_{k+1}\}).
\end{equation}

Case 2.2.1 (step C2.3.7): $\pset{C}$ is abnormal, $\deg(I_{k+1}, \lv(I_{k+1})) < \deg(C_l, \lv(I_{k+1}))$, and $\prem(\ini(Q), $ $\pset{C}_{l-1}) = 0$. By the formula $L^q C_l = Q I_{k+1} + R$ of pseudo-division of $C_l$ with respect to $I_{k+1}$, where $q\in \znum_{\geq 0}$ and $Q, R \in \kx$, we have $\ini(Q) = L^{q-1} I_l$. Since
\begin{equation*}
\prem(\ini(Q), \pset{C}_{l-1}) = \prem(L^{q-1} I_l, \pset{C}_{l-1}) =  I_l\prem(L^{q-1}, \pset{C}_{l-1}) = 0,
\end{equation*}
$\prem(L^{q-1}, \pset{C}_{l-1}) = 0$ and thus
\begin{equation}\label{eq:zero3}
\zero(\pset{P}) = \bigcup \limits _{i=1}^{l-1}\zero(\pset{G} \cup \{I_i\}) \cup \zero(\pset{G} \cup \{L\})
\end{equation}
follows from arguments similar to those in case 2.1.

Case 2.2.2 (step C2.3.8): $\pset{C}$ is abnormal, $\deg(I_{k+1},\lv(I_{k+1})) < \deg(C_l,  \lv(I_{k+1}))$, and $\prem(\ini(Q), $ $\pset{C}_{l-1})\neq 0$. By Theorem~\ref{thm:divisible}(b), we have
\begin{equation*}
  \begin{split}
\prem(C_l, [C_1, \ldots, C_{l-1}, I_{k+1}]) &=\prem(\prem(C_l, I_{k+1}), \pset{C}_{l-1}) \\
&= \prem(L^q C_l - Q I_{k+1}, \pset{C}_{l-1}) = 0,
  \end{split}
\end{equation*}
where $Q$ and $q$ are as in case 3. Then there exist $q_1, \ldots, q_{l-1}\in \znum_{\geq 0}$ and $Q_1, \ldots, Q_{l-1} \in \kx$ such that $I_1^{q_1} \cdots I_{l-1}^{q_{l-1}}(L^q C_l - Q I_{k+1}) = \sum_{i=1}^{l-1} Q_i C_i$. It follows that
\begin{equation}
  \label{eq:proof-split4}
-I_1^{q_1} \cdots I_{l-1}^{q_{l-1}} I_{k+1} Q= \sum \limits _{i=1}^l Q_i C_i
\end{equation}
for $Q_l = -I_1^{q_1} \cdots I_{l-1}^{q_{l-1}}L^q$. Let the formula of pseudo-division of $Q$ with respect to $\pset{C}_{l-1}$ be
\begin{equation}
  \label{eq:proof-split4-2}
I_1^{\bar{q}_1} \cdots I_{l-1}^{\bar{q}_{l-1}} Q = \sum \limits _{i=1}^{l-1}\bar{Q}_iC_i + \prem(Q, \pset{C}_{l-1})
\end{equation}
for some $\bar{q}_1, \ldots, \bar{q}_{l-1}\in \znum_{\geq 0}$ and $\bar{Q}_1$, $\ldots, \bar{Q}_{l-1} \in \kx$. Then one can find $\hat{q}_1, \ldots, \hat{q}_{l-1}\in \znum_{\geq 0}$ and $\hat{Q}_1$, $\ldots, \hat{Q}_l \in \kx$ such that
$$I_1^{\hat{q}_1} \cdots I_{l-1}^{\hat{q}_{l-1}}I_{k+1} \prem(Q, \pset{C}_{l-1})  = \sum \limits _{i=1}^l \hat{Q}_i C_i$$
holds (in view of the formulas \eqref{eq:proof-split4} and \eqref{eq:proof-split4-2}). Therefore,
\begin{equation}\label{eq:zero4}
\zero(\pset{P}) = \bigcup \limits _{i=1}^{l-1}\zero(\pset{G} \cup \{I_i\}) \cup \zero(\pset{G} \cup \{I_{k+1}\}) \cup \zero(\pset{G} \cup \{\prem(Q, \pset{C}_{l-1})\}).
\end{equation}

The zero relations in \eqref{eq:zero1}, \eqref{eq:zero2}, \eqref{eq:zero3}, and \eqref{eq:zero4} show that for each polynomial set $\pset{P}\in \Phi$, any zero of $\pset{P}$ is either in $\zero(\pset{C}/\ini(\pset{C}))$ if the W-characteristic set $\pset{C}$ of $\bases{\pset{P}}$ is normal in case 1 or in $\zero(\pset{P}')$ for another polynomial set $\pset{P}'$ adjoined to $\Phi$ for later computation in the other cases. This proves the zero relation $\zero(\pset{F}) = \bigcup_{(\pset{G}, \pset{C}) \in \Psi}\zero(\pset{C} / \ini(\pset{C}))$; for Algorithm~1 terminates when $\Phi$ becomes empty.

On one hand, by the zero relation \eqref{eq:zero}, we have
$$\zero(\pset{F}) = \bigcup_{(\pset{G}, \pset{C}) \in \Psi}\zero(\pset{C} / \ini(\pset{C}))  \subseteq \bigcup_{(\pset{G}, \pset{C})\in\Psi} \zero(\pset{G}).$$
On the other hand, $\zero(\pset{G}) \subseteq \zero(\pset{F})$ holds for all $(\pset{G}, \pset{C})\in \Psi$ according to the zero relations \eqref{eq:zero1}, \eqref{eq:zero2}, \eqref{eq:zero3}, and \eqref{eq:zero4} for the four cases of splitting. This proves the equality $\zero(\pset{F}) = \bigcup_{(\pset{G}, \pset{C})\in\Psi} \zero(\pset{G})$. For each $(\pset{G}, \pset{C}) \in \Psi$, we have $\zero(\pset{C}/\ini(\pset{C})) \subseteq \zero(\sat(C)) \subseteq \zero(G)$, and thus
$$\zero(\pset{F}) = \bigcup_{(\pset{G}, \pset{C})\in\Psi} \zero(\pset{C}/\ini(\pset{C})) \subseteq \bigcup_{(\pset{G}, \pset{C})\in\Psi} \zero(\sat(\pset{C})) \subseteq \bigcup_{(\pset{G}, \pset{C})\in\Psi} \zero(\pset{G}) = \zero(\pset{F}).$$
This completes the proof of the zero relation \eqref{eq:NormalDec}.
\end{proof}\vspace{-5mm}

\section{Example and experiments}\label{sec:ex-ex}

\subsection{Example for characteristic decomposition}

Let $\pset{F} = \{ay-x-1, -xyz+az, xz^2-az+y\} \subseteq \fk[a, x, y, z]$ with $a<x<y<z$. The procedure to compute a characteristic decomposition of $\pset{F}$ using Algorithm~1 is shown in Table~\ref{tab:charpair}, where $\pset{G}_i$ is the computed reduced lex \grobner basis and $\pset{C}_i$ is its W-characteristic set in the $i$th loop.

\begin{table*}[!htp]
\centering
\caption{Illustration for Algorithm~1}\label{tab:charpair}

\medskip
\resizebox{\textwidth}{!}{
\begin{tabular}{c|ccccccc}
$i$ & $\Phi$ & $\pset{P}$ & $\pset{G}_i$ & $\pset{C}_i$ & Normal & Case &$\Psi$\\ ~\\[-12pt]
\hline
1 & $\{\pset{F}\}$ & $\pset{F}$ &
\begin{tabular}{l}
$\{G_1,G_2,G_3,G_4$,\\
~\,$zG_5,zG_6,G_7\}$
\end{tabular}
 & $[G_1,G_2,zG_5]$ & No & 2.2.2 & $\emptyset$  \\
\hline\vspace{-1.5mm}
~\\
2 & $\{\pset{F}_1,\pset{F}_2\}$ & $\pset{F}_1$ & $\{x+1,y,az,z^2\}$& $[x+1,y,az]$ & Yes & 1 &$\{(\pset{G}_2,\pset{C}_2)\}$ \\[2.5mm]
\hline
 3 & $\{\pset{F}_2,\pset{F}_3\}$ & $\pset{F}_2$&
 \begin{tabular}{c}
$\{G_2,G_4,G_5,$\\
$zG_6,G_7\}$
 \end{tabular}
& $[G_2, G_5, zG_6]$ & No & 2.1 &$\{(\pset{G}_2,\pset{C}_2)\}$\\
\hline\vspace{-1.5mm}
~\\
4 & $\{\pset{F}_3,\pset{F}_4,\pset{F}_5\}$ & $\pset{F}_3$ & $\{a,x+1,y,z\}$& $[a,x+1,y,z]$ & Yes & 1 &
\begin{tabular}{c}
$\{(\pset{G}_2,\pset{C}_2), (\pset{G}_4,\pset{C}_4)\}$
\end{tabular}\\[2.5mm]
\hline
5 & $\{\pset{F}_4,\pset{F}_5\}$ & $\pset{F}_4$ & $\{G_5, G_2,G_6,G_7\}$ & $[G_5,G_2,G_7]$ & Yes & 1 & \begin{tabular}{c}
$\{(\pset{G}_2,\pset{C}_2)$, $(\pset{G}_4,\pset{C}_4)$\\
$(\pset{G}_5, \pset{C}_5)\}$
\end{tabular}\\
\hline
6 & $\{\pset{F}_5,\pset{F}_6\}$ & $\pset{F}_5$ &
\begin{tabular}{c}$\{a, x+1, y^2$,\\$ yz, z^2-y\}$\end{tabular}
& $[a,x+1,y^2,yz]$ & No & 2.2.2 & \begin{tabular}{c}
$(\{\pset{G}_2,\pset{C}_2)$, $(\pset{G}_4,\pset{C}_4)$\\
$(\pset{G}_5, \pset{C}_5)\}$
\end{tabular}\\
\hline
7 & $\{\pset{F}_6\}$ & $\pset{F}_6$ & $\{a,x+1,y,z^2\}$ & $[a,x+1,y,z^2]$ & Yes & 1 & \begin{tabular}{l}
\{$(\pset{G}_2,\pset{C}_2)$, $(\pset{G}_4,\pset{C}_4)$\\
$~\,(\pset{G}_5, \pset{C}_5)$, $(\pset{G}_7, \pset{C}_7)$\}
\end{tabular}\\
\end{tabular}}
\end{table*}

The polynomial sets $\pset{F}_i$ and polynomials $G_j$ in Table~\ref{tab:charpair} are
listed below:
\begin{equation*}
  \begin{array}{ll}\smallskip
     \pset{F}_1 = \{x+1, y, ay, az, (y+a)z, G_7\}, & \pset{F}_2 = \{G_4,aG_2, zG_6, G_2, G_7, G_5\},\\ \smallskip
      \pset{F}_3 = \{a,x+1,y,z\}, &  \pset{F}_4 = \{G_5, G_2, G_6, G_7\},\\ \medskip
      \pset{F}_5 = \{a, x+1, x^2+x, G_4, xyz, G_7\}, &  \pset{F}_6 = \{a, x+1, x^2+x, xy, G_7\}; \\ \smallskip
      G_1 = x^3+2x^2+(1-a^2)x-a^2, &  G_2= ay-x-1,\\ \smallskip
      G_3=x^2y+xy-ax-a, & G_4=xy^2-x-1, \\ \smallskip
      G_5=x^2+x-a^2, & G_6=xy-a, \\
      G_7=z^2-yz+y^3-y.
  \end{array}
\end{equation*}
The output characteristic decomposition of $\pset{P}$, as shown in Table~\ref{tab:charpair}, is $\{(\pset{G}_2, \pset{C}_2), (\pset{G}_4, \pset{C}_4)$, $(\pset{G}_5, \pset{C}_5), (\pset{G}_7, \pset{C}_7)\}$.

\subsection{Experimental results}

Algorithm~1 has been implemented in {\sc Maple} 17 based on functions available in the {\sf FGb} and {\sc Maple}'s built-in packages for \grobner basis computation. The implementation will be included in the upcoming new version of the Epsilon package for triangular decomposition \cite{W2002e}.

In Theorem~\ref{thm:divisible} there is an assumption on the variable order (i.e., all the parameters of a W-characteristic set are ordered before the other variables). This assumption always holds in the zero-dimensional case. Our experiments show that in the positive-dimensional case there are about one fourth of the test examples for which it happens that the assumption does not hold. The assumption is made to ensure that the (pseudo-) divisibility relationships in Theorem~\ref{thm:divisible} occur. In the case where the assumption does not hold, we can make such relationships to occur by changing the variable order heuristically. In fact, using the heuristics we were able to obtain necessary (pseudo-) divisibility relationships to complete the characteristic decomposition for all the test examples.

Let us emphasize that Algorithm~1 decomposes any polynomial set into characteristic pairs of reduced lex \grobner bases and their W-characteristic sets at one stroke. The two kinds of objects resulted from the combined decomposition, each having its own structures and properties, are interconnected. This makes our algorithm distinct from other existing ones for triangular decomposition.
To observe the computational performance of the algorithm, in comparison with algorithms for indirect normal decomposition (that is, first computing a regular decomposition and then normalizing the regular sets in the decomposition), we made some experiments on an Intel(R) Core(TM) Quad CPU at 2.83 GHz with 4.00 GB RAM under Windows 7 Home Basic. Selected results of the experiments are presented in Table~\ref{tab:main}, of which the first 9 are taken from \cite{c07c} and the others are from benchmarks for the {\sf FGb} library. We implemented Algorithm~1 as $\algnor$ in {\sc Maple} for characteristic decomposition and used the functions {\sf Triangularize} (from the {\sf RegularChains} package in {\sc Maple}) and {\sf RegSer} (from the {\sf Epsilon} package for {\sc Maple}) for regular decomposition and the function {\sf normat} (from the {\sf miscel} module of {\sf Epsilon}) for normalization.

 In Table~\ref{tab:main}, ``Source" indicates the label in the above-cited references and ``Dim" denotes the dimension of the ideal in the example. ``Total" under $\algnor$ records the total time (followed by the number of pairs in parenthesis) for characteristic decomposition using Algorithm~1; ``GB" under $\algnor$ records the time for computing all the reduced lex \grobner bases; ``Total" and ``Regular" under {\sf RegSer} and {\sf Triangularize} record the total time for normal decomposition and the time for regular decomposition (followed by the numbers of components in parenthesis) respectively, where normal decompositions are computed from regular decompositions by means of normalization using {\sf normat}. The marks ``lost" and ``$>4000$" in the columns mean that Maple reports ``lost kernel connections'' and that the computation does not terminate within 4000 seconds respectively. 

\begin{table*}[!t]
\centering
\caption{Timings for characteristic/normal decomposition}\label{tab:main}

\medskip
{\footnotesize
\begin{tabular}{c|c|cc|cc|cc}
\multicolumn{2}{c}{ } & \multicolumn{2}{c}{\sf CharDec}& \multicolumn{2}{c}{\sf RegSer}&\multicolumn{2}{c}{\sf Triangularize}\\
\hline
Source & Dim   & Total  & GB   & Total & Regular & Total  & Regular\\
\hline
S5  &4 &0.14(8)  &0.047       &4.182(31)     &0.484(19)    &1.513(9)  &0.124(1)  \\
S7  &1 &0.156(5)  &0.078       &0.251(7)     &0.11(4)    &0.249(5)  &0.109(1)  \\
S8  &2 &0.062(2)  &0.32       &0.062(3)     &0.062(3)    &0.156(2)  &0.141(2)  \\
S9  &2 &0.125(5)  &0.078       &0.483(21)     &0.14(8)    &0.188(6)  &0.094(1)  \\
S10 &3& 0.594(16)  &0.313       &0.438(7)     &0.172(5)    &0.36(3)  &0.235(1)  \\
S13  &3& 0.312(13)  &0.14       &0.171(8)     &0.109(8)    &0.125(2)  &0.094(1)  \\
S14  &2 & 0.531(9)  &0.327       &0.14(6)     &0.109(6)    &0.157(8)  &0.125(8)  \\
S16  &3& 0.640(6)  &0.344       &0.703(7)     &0.609(7)    &4.609(8)  &4.609(8)  \\
S17  &6 & lost   &lost       &lost     &lost    &lost  &lost  \\
nueral    &1 &1.826(15)      &1.514      &$>4000$     &$>4000$    & 0.233(6)       &0.14(5)         \\
F663  & 2 &2.949(6)      &2.326    &1.935(16) &1.202(15)&1.607(6)  &1.045(4)  \\
Dessin2  &0 &27.222(1)     &27.207  &$>4000$     &$>4000$    &$>4000$     &$>4000$      \\
Wang16  &0 &0.203(1)      &0.171          &14.555(1) &0.437(1)&14.086(1) &0.156(1)  \\
filter9 &0 &0.640(1)      &0.593          &$>4000$   &$>4000$ &lost      &lost\\
fabrice24 &0 &436.7(1)    &436.7          &lost      &lost    &lost      &lost\\
uteshev bikker &0 &3.806(1) &3.766        &lost      &lost    &$>4000$   &$>4000$\\
Cyclic6   &0 &2.153(25)      &1.244          &lost     &lost    &$>4000$ &$>4000$  \\
\hline
\end{tabular}}
\end{table*}

The most time-consuming step in $\algnor$ is for the computation of lex \grobner bases, as one can see from Table~\ref{tab:main}. In our implementation, the {\sf FGb} library is first invoked to compute \grobner bases with respect to graded reverse lexicographic term ordering, and the computed \grobner bases are then converted to lex \grobner bases by changing the term ordering using either the FGLM algorithm for the zero-dimensional case \cite{FGLM93E} or the \grobner walk otherwise \cite{CKM97C}. Unfortunately, the built-in implementation of the \grobner walk algorithm in {\sc Maple} is very inefficient and it is the current bottleneck of our implementation. 

~\\
~\\
Finally, we add a few remarks to conclude the paper: we have studied characteristic pairs (that is, pairs of reduced lex \grobner bases and normal triangular sets) and the problem of characteristic decomposition (that is, decomposition of arbitrary polynomial sets into characteristic pairs). We have proved a number of properties about characteristic pairs and characteristic decomposition, and proposed an algorithm with implementation for the decomposition. The algorithm explores the inherent connection between Ritt characteristic sets and lex \grobner bases and involves mainly the computation of lex \grobner bases; normal triangular sets are obtained as by-product almost for free. Associated to a characteristic decomposition $\{(\pset{G}_1, \pset{C}_1), \ldots, (\pset{G}_t, \pset{C}_t)\}$ of a polynomial set $\pset{P}$ are zero decompositions
$$\zero(\pset{P}) = \zero(\pset{G}_1) \cup \cdots \cup \zero(\pset{G}_t) =  \zero(\pset{C}_1/\ini(\pset{C}_t))\cup \cdots \cup \zero(\pset{C}_t/ \ini(\pset{C}_t))$$
and the corresponding radical ideal decompositions
$$\sqrt{\bases{\pset{P}}} = \sqrt{\bases{\pset{G}_1}}\cup \cdots \cup \sqrt{\bases{\pset{G}_t}} = \sqrt{\sat(\pset{C}_1)}\cup \cdots \cup \sqrt{\sat(\pset{C}_t)}.$$
In these decompositions, the reduced lex \grobner bases $\pset{G}_1, \ldots, \pset{G}_t$ and normal triangular sets $\pset{C}_1, \ldots, \pset{C}_t$ are closely linked and well structured polynomial sets whose usefulness has been widely recognized.

{\small
\bibliographystyle{abbrv}
\bibliography{alg4writt}
}

\end{document}